\documentclass[journal]{IEEEtran} 
\usepackage{bbm} % For sets and Expectation. Double-striked left side and normal right side fonts

\usepackage{bm} % For Vectors, Matrices. Command \bm

\usepackage{amsmath} % American Math Society Math

\usepackage{amsthm} % Theorem Like Definition

\usepackage{amssymb} % Extended symbol collection

\usepackage{mathrsfs} % For the command \mathscr - Ralph Smith's Formal Script font

\usepackage{optidef}  % Standard set of environments for writing optimization problems.

\usepackage{algorithm} % Algorithm

\usepackage[noend]{algpseudocode} % Pseudocode

\usepackage{cite} % Cite

\usepackage{graphicx} % Inserting images

\usepackage{textcomp} % Special Symbols

\usepackage{xcolor} % Colourful text

\usepackage{caption} % Caption
\usepackage{subcaption} % Sub Captions

\usepackage{hyperref} % Clickable References

\usepackage{footnote}

\newtheorem{theorem}{Theorem}
\newtheorem{lemma}{Lemma}

\newtheorem{remark}{Remark}

\graphicspath{ {./Figures/} }

\DeclareMathOperator*{\argmin}{arg\,min}

\begin{document}

\title{\spaceskip=4.2pt An Efficient Reservation Protocol for Medium Access: When Tree Splitting Meets Reinforcement Learning}

\author{Yutao Chen and Wei Chen
    \thanks{Yutao Chen and Wei Chen are with the Department of Electronic Engineering, Tsinghua University, Beijing 100084, China, and also with the State Key Laboratory of Space Network and Communications, as well as, the Beijing National Research Center for Information Science and Technology (email: cheny1995@tsinghua.edu.cn, wchen@tsinghua.edu.cn).}
}

\maketitle

\begin{abstract}
    As an enhanced version of massive machine-type communication in 5G, massive communication has emerged as one of the six usage scenarios anticipated for 6G, owing to its potential in industrial internet-of-things and smart metering. Driven by the need for random multiple-access (RMA) in massive communication, as well as, next-generation Wi-Fi, medium access control has attracted considerable recent attention. Holding the promise of attaining bandwidth-efficient collision resolution, multiaccess reservation no doubt plays a central role in RMA, e.g., the distributed coordination function (DCF) in IEEE 802.11. In this paper, we are interested in maximizing the bandwidth efficiency of reservation protocols for RMA under quality-of-service constraints. Particularly, we present a tree splitting based reservation scheme, in which the attempting probability is dynamically optimized by partially observable Markov decision process or reinforcement learning (RL). The RL-empowered tree-splitting algorithm guarantees that all these terminals with backlogged packets at the beginning of a contention cycle can be scheduled, thereby providing a first-in-first-out service. More importantly, it substantially reduces the reservation bandwidth determined by the communication complexity of DCF, through judiciously conceived coding and interaction for exchanging information required by distributed ordering. Simulations demonstrate that the proposed algorithm outperforms the CSMA/CA based DCF in IEEE 802.11.
\end{abstract}

\vspace{-2em}
\section{Introduction}
% General introduction to RMA and Fixed allocation approach
With the release of the IMT-2030 (6G) framework, researchers have intensified efforts to advance 6G technologies~\cite{letaief2019roadmap}. The IMT-2030 framework outlines six usage scenarios for 6G, including the massive communication~\cite{guo2021enabling}, an evolution of 5G's massive machine-type communication (mMTC)~\cite{bockelmann2016massive}. This scenario envisions ubiquitous connectivity for a vast number of devices and sensors, enabling applications such as the industrial Internet of things (IIoT), smart cities, healthcare, and energy systems. To meet the stringent quality-of-service (QoS) requirements of massive communication, efficient multiplexing and multiaccess techniques have received renewed attention~\cite{mao2022rate,liu2022evolution,hui2022delay,hui2023fresh}. One of the most important approaches is the efficient allocation of communication resources, with a major research direction focusing on fixed allocation techniques~\cite{sklar2021digital}, such as frequency-division multiple access (FDMA), time-division multiple access (TDMA), and code-division multiple access (CDMA). These methods allocate resources statically, prior to service requests. While effective for predictable traffic, this fixed allocation leads to resource underutilization when service demands are bursty or intermittent. In contrast, dynamic allocation adapts to fluctuating service demands. For example, under bursty service requests, fixed allocation often leads to many unused resources, whereas dynamic allocation efficiently distributes resources based on real-time demand. Therefore, access algorithms that enable dynamic allocation are essential for optimizing resource management.

% The beginning of on-demand solution
Access algorithms primarily operate at the medium access control (MAC) layer. Early approaches include ALOHA~\cite{abramson1970aloha}, where terminals retransmit after a random delay if a collision occurs. A variant called slotted ALOHA~\cite{roberts1975aloha} was later proposed, where time is globally synchronized and divided into slots, and transmission occurs only at the beginning of a time slot. Another important class of access algorithms is the tree algorithm~\cite{capetanakis2003tree, tsybakov1980random}, where terminals are assigned hierarchical levels based on past transmission outcomes, with subsequent decisions made according to the assigned level. Systematic overviews of ALOHA-like and tree-based algorithms can be found in~\cite{tsybakov1985survey,bertsekas2021data}. The performance of the above algorithms and the theoretical limits of such multiple access algorithms have been widely studied~\cite{pippenger1981bounds, molle1982capacity, cruz1982new, fayolle1985analysis,liew2009bounded, dai2012stability}. However, the direct data transmission mechanism can lead to excessive resource waste during collisions, especially when data packets are large. Therefore, researchers have explored alternatives such as carrier sense multiple access (CSMA)~\cite{kleinrock1975packet, bianchi2000performance}, where terminals can detect whether the channel is busy for a short period before transmitting. Leveraging this information, terminals wait when the channel is busy to avoid collisions and initiate transmission as soon as the channel becomes idle. However, collisions can still occur when multiple terminals detect the idle channel simultaneously. To further reduce collisions during data transmission, reservation-based protocols~\cite{sklar2021digital} have been developed. These protocols use short reservation packets to reserve the communication channel before data transmission. This process, known as the contention cycle, improves performance by separating contention resolution from data transmission. For time-sensitive applications where data must be transmitted within a strict time frame, polling algorithms~\cite{towsley1984adaptive} provide a structured approach by having a central terminal periodically poll devices to determine transmission needs. Beyond collision avoidance, random multiple-access (RMA) has also been investigated through interdisciplinary perspectives.

% Other approachs to on-demand solution
One prominent approach is information-theoretic RMA that integrates information theory, which traditionally ignores random packet arrivals, and networking, which traditionally lacks detailed physical-layer modeling~\cite{gallager1985,ephremides1998}. For instance, the authors in~\cite{minero2012} leverage the adaptive coding schemes to improve throughput. The random access problem is also formulated as a coding problem in~\cite{polyanskiy2017}. Under this formulation, \cite{paolini2015} introduces coded slotted ALOHA (CAS), which combines packet erasure correcting codes and successive interference cancellation (SIC). A systematic review of coding-theoretic RMA approaches can be found in~\cite{liva2024}. Another key direction applies game theory to model the RMA problem as a non-cooperative game, where terminals make independent access decisions based on their strategies and payoffs. This approach balances fairness and efficiency in channel access~\cite{lee2007Reverse, sagduyu2009game, chen2010}. The RMA problem has also been examined as a network utility maximization (NUM) problem~\cite{nandagopal2000, chen2005, lee2007Utility}, emphasizing global fairness. While these approaches have advanced RMA optimization, communication complexity~\cite{yao1979some} remains high, particularly under heavy traffic conditions. Hence, re-examining RMA with emerging technologies, such as artificial intelligence (AI), shows great potential.

% Reinforcement Learning POMDP (connects to our approach)
IMT-2030 (6G) identifies AI and communication as one of its six usage scenarios. Given that decision-making lies at the core of access algorithms and the abundance of computing power, we explore the use of reinforcement learning (RL)~\cite{russell2016artificial} to enhance the performance of the access algorithm. Within the RL framework, agents learn to optimize decisions based on feedback from their environment. Among various RL models, the Markov decision process (MDP) is fundamental for fully observable environments, while the partially observable Markov decision process (POMDP)~\cite{spaan2012partially} extends this framework to scenarios where agents act under incomplete state information. In the RMA problem, terminals must make transmission decisions without direct communication with each other, meaning no single terminal possesses complete system knowledge. This naturally aligns with the POMDP framework. To solve POMDPs, deep learning based algorithms can be applied to large-scale POMDPs~\cite{hausknecht2015deep,igl2018deep} and value iteration algorithms (VIA) designed for MDP can be adapted to solve small-to-medium sized POMDPs. For example, the point-based value iteration (PBVI)~\cite{pineau2003point} updates only a finite set of representative belief points, while the real-time dynamic programming with belief state (RTDP-Bel)~\cite{bonet2000} focuses on updates for high-probability states. In this paper, we formulate the RMA problem as a POMDP, aiming to synergize tree-splitting efficiency with RL-driven learning capabilities to develop an efficient reservation protocol.

% Contribution
The key contributions and novelties of this paper are summarized as follows.
\begin{itemize}
    \item This paper investigates the RMA problem in a setting where a reservation protocol governs channel reservation before data transmission. Leveraging the structural properties of the problem, we formulate the RMA problem as a POMDP. Based on this formulation, we propose a reservation protocol that optimizes the reservation process through carefully designed interaction mechanisms for information exchange. The reservation protocol is presented and optimized for the access stage in both wired and wireless local area networks (LANs). In particular, whether an ACK or NACK is feedback given the number of simultaneously attempting terminals, is the same as that assumed in wired LANs, such as Ethernet~\cite{metcalfe1976ethernet} and bus topology~\cite{forouzan2007data}, as well as in ALOHA networks and the IEEE 802.11 DCF for wireless LANs (WLANs). As a result, the conceived protocol are compatible with the current hardware platform of IEEE 802.11, thereby being capable of replacing existing reservation schemes of IEEE 802.11 DCF without much hardware upgrades.
    \item We consider a class of reservation protocols where terminals with identical information adopt identical probabilistic actions. Within this class, we propose an efficient reservation protocol that is learned through real-time dynamic programming (RTDP). Additionally, we introduce a pre-training strategy that significantly accelerates the learning. The proposed reservation protocol also ensures first-in-first-out (FIFO) service, guaranteeing that packets are transmitted in the order of their arrival contention cycles. This structured transmission facilitates a more tractable performance analysis.
    \item Compared to classic random access protocols (e.g., slotted ALOHA) and modern random access protocols (e.g., the CSMA/CA-based DCF in IEEE 802.11), the proposed reservation protocol achieves higher effective throughput, particularly under heavy traffic loads, due to its shorter reservation process. Extensive numerical simulations confirm the efficiency of the proposed approach.
\end{itemize}

% Structure
The remainder of this paper is organized as follows. Section~\ref{sec:SystemModel} introduces the system model and an overview of the reservation protocol. In Section~\ref{sec:POMDP}, we present the POMDP-based optimization approach for the reservation protocol, followed by a genie-aided pre-training strategy to accelerate the learning in Section~\ref{sec:ReservationProtocolImplemantation}. The paper concludes with Section~\ref{sec:NumericalResults}, providing simulation results to highlight the performance of the proposed reservation protocol.

\section{System Model}\label{sec:SystemModel}
We consider a slotted-time system where multiple terminals attempt to transmit data packets over a shared communication channel without any assigned priority. Due to limited communication resources, only one terminal can successfully transmit at any given time.

\subsection{Reservation-Based Random Multiple Access}
To facilitate efficient channel access, the terminals employ a reservation protocol. To this end, time is divided into consecutive frames, with the channel reservation process initiated only at the beginning of each frame. The frame consists of multiple time slots allocated for both channel reservation and packet transmission. We assume that packet arrivals are stochastic and that only terminals with backlogged packets at the start of a frame participate in the reservation process. Terminals in the reservation process are referred to as active terminals. Let $N_k\in\mathbbm{N}^0$ denote the number of active terminals at the beginning of frame $k$. Given the stochastic nature of data arrivals, $N_k$ is a random variable whose distribution depends on the frame division strategies. In this paper, we consider the following two time-division-based strategies.
\begin{itemize}
    \item \textbf{Fixed frame}: Each frame has a predetermined duration $T$, and the channel reservation process does not occupy the entire frame. Active terminals initiate data packet transmissions after the reservation process concludes, with backlogged packets from previous frames transmitted first, following the order established in those frames.
    \item \textbf{Dynamic frame}: A new frame begins only after all active terminals in the current frame have successfully transmitted their data packets.
\end{itemize}
Terminals send a dedicated finish signal to indicate the end of the packet transmission. Through this mechanism, terminals can determine the number of terminals in the transmission queue. Specifically, terminal entries are dictated by channel feedback during the reservation, while departures are signaled by the finish signal. This ensures that terminals are aware of when to begin their transmissions and when a new frame starts.
\begin{figure}[t]
    \centering
    \includegraphics[width=\columnwidth]{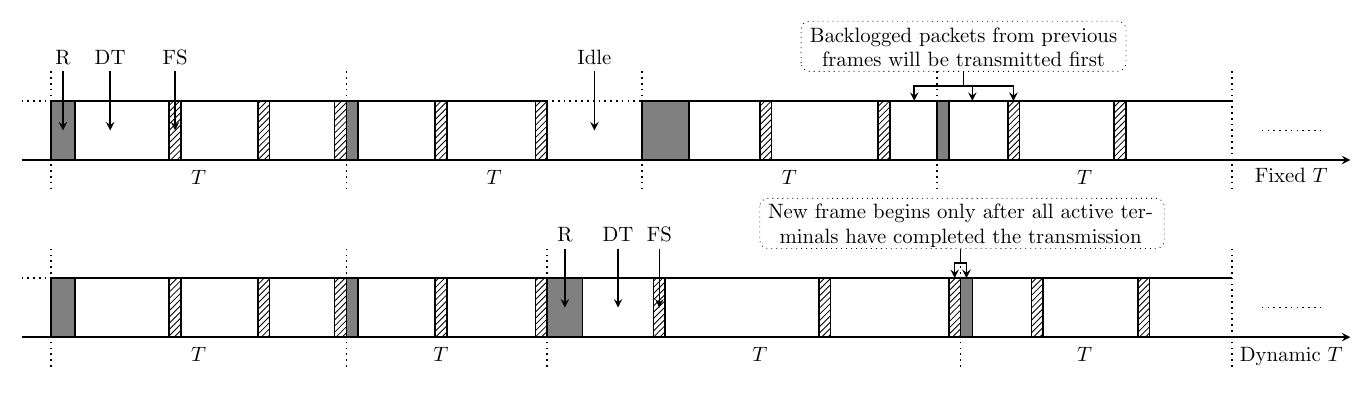}
    \caption{Illustration of the frame division strategies, where R, DT, and FS denote reservation, data transmission, and finish signal, respectively.}
    \label{fig:FrameDivision}
    \vspace{-1.5em}
\end{figure}
An illustration of the two strategies is provided in Fig.~\ref{fig:FrameDivision}. For the remainder of this paper, we assume the distribution of $N_k$ is known prior to each channel reservation process. It can be noted that we are interested in the access stage of the reservation-based multiaccess communications. In wireless environments, the physical-layer settings can be treated in the same manner as those in the IEEE 802.11 DCF, thereby allowing related issues to be handled similarly. In this context, the reservation process can be modeled from an information-theoretic perspective, as detailed in Section~\ref{sec:InformationTheorticModel}.

\subsection{Information-Theoretic Perspective on Reservation Process}\label{sec:InformationTheorticModel}
At the beginning of each frame, the active terminals compete for channel access through multiple rounds of reservation packet transmissions. The reservation packets may either contain a fraction of the data packet or serve solely as a reservation request. During the channel reservation phase, the system operates synchronously, providing error-free channel feedback $c_t$ in each time slot $t$. Specifically, $c_t=0$ denotes no transmission, $c_t=1$ indicates that a single terminal has transmitted a reservation packet, and $c_t=e$ signifies a collision. When $c_t=1$, the transmitting terminal wins the reservation and becomes inactive, while the remaining active terminals continue contending for channel access. Let $A_{i,t}$ be the binary transmission decision of active terminal $i$ at time slot $t$, where $A_{i,t}=1$ indicates the transmission of a reservation packet and $A_{i,t}=0$ otherwise. Then, we have
\begin{equation}
    c_t = \bigoplus_{i=1}^{N} A_{i,t} \triangleq \begin{cases}
        0 & \sum_{i=1}^{N}A_{i,t} = 0 \\
        1 & \sum_{i=1}^{N}A_{i,t} = 1 \\
        e & \text{otherwise},
    \end{cases}
\end{equation}
where $N$ is the number of active terminals. We assume that each reservation round, consisting of both the transmission of reservation packets and the broadcasting of channel feedback, occupies a single time slot.

The reservation protocol governs channel reservation by determining each active terminal's probability of transmitting the reservation packets in each time slot. Let $\theta_{i,t}$ be the status of active terminal $i$ at time slot $t$. The reservation protocol assigns the transmission probability for each active terminal based on its individual status $\theta_{i,t}$, which is updated according to the transmission decision and the corresponding channel feedback. A schematic representation of the active terminal's status evolution is provided in Fig.~\ref{fig:ChannelModel}.
\begin{figure}[t]
    \centering
    \includegraphics[width=\columnwidth]{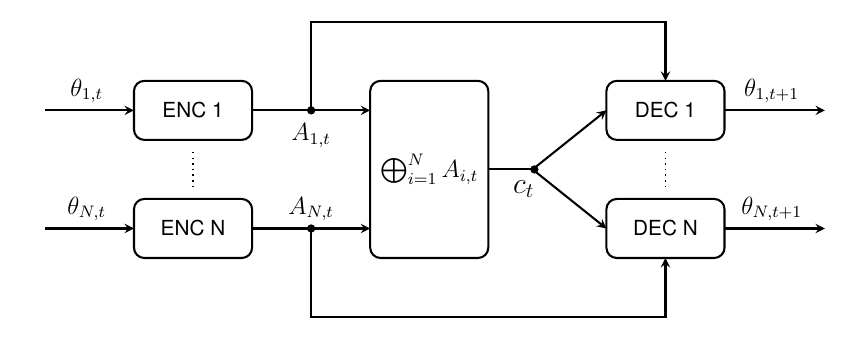}
    \caption{Channel model representing active terminals' status evolution. The decision-making and status update are modeled as an encoder (ENC) and a decoder (DEC), respectively.}
    \label{fig:ChannelModel}
    \vspace{-1.5em}
\end{figure}

\subsection{Tree-Splitting-Based Reservation}\label{sec:ReservationProtocol}
In this subsection, we conceive a tree-splitting-based reservation protocol, under which active terminals with identical decision histories are assigned the same transmission probability. As a result, active terminals can be grouped into clusters based on their decision histories. Then, the reservation protocol can be characterized by a vector $\bm{p}_t=[p_{t,1},p_{t,2},...,p_{t,M_t}]$, where $0\leq p_{t,i}\leq 1$ represents the transmission probability for the terminals in cluster $i$, and $M_t$ denotes the total number of clusters at time slot $t$. Initially, all active terminals belong to a single cluster.
\begin{remark}
    According to~\cite{tsybakov1985survey}, a reservation protocol is defined as a function of (1) the generation time of the data packet, (2) the channel feedback, and (3) the terminal's decision history. Since no priority is assigned among terminals in this paper, the data packet generation time does not affect the reservation protocol. Moreover, as channel feedback is broadcast to all terminals without error, the reservation protocol relies solely on the terminal's decision history. Therefore, we focus on the reservation protocol outlined above.
\end{remark}

To detail the operation of the tree-splitting-based reservation protocol, we first characterize the evolution of $M_t$. To this end, we consider a typical time slot in which a collision occurs. In this case, the active terminals involved in the collision are relocated to newly generated clusters. However, to ensure that active terminals can reliably identify their cluster membership, which is an essential requirement for acting according to the reservation protocol, only one new cluster is created per collision, even if the colliding reservation packets originate from terminals in different clusters with distinct decision histories. When no transmission occurs, or when an active terminal wins the reservation, the active terminals retain their respective cluster memberships. Consequently, the evolution of $M_t$ is captured by
\begin{equation}\label{eq:clusterNumberEvolution}
    M_{t+1} = \begin{cases}
        M_t     & c_t \neq e \\
        M_t + 1 & c_t = e.
    \end{cases}
\end{equation}
It is important to note that the number of terminals in each cluster is non-increasing and some clusters may become empty during the reservation process. Additionally, the evolution of $M_t$ depends solely on the channel feedback.

The status of an active terminal at time slot $t$ can be characterized by $\theta_t\triangleq [M_t,j_t,\Xi_t]$, where $j_t$ is the cluster index to which the terminal belongs and $\Xi_t$ represents the information about the other terminals. The details of $\Xi_t$ and its evolution will be examined in subsequent sections. Then, the operation of an active terminal under the reservation protocol proceeds as follows.
\begin{enumerate}
    \item The terminal initializes $\theta_1 = [1,1,\Xi_1]$.
    \item At time slot $t$, the terminal transmits the reservation packet with probability $p_{t,j_t}$, as specified by the reservation protocol.
    \item Based on the transmission decision and channel feedback, the terminal operates as follows.
          \begin{itemize}
              \item When $A_t=1$ and $c_t = 1$, the terminal wins the competition and becomes inactive.
              \item When $A_t=1$ and $c_t = e$, the terminal updates $\theta_{t+1} = [M_t+1,M_t+1,\Xi_{t+1}(e)]$, where $\Xi_{t+1}(c_t)$ is the updated information about other terminals when $c_t$ is received.
              \item When $A_t=0$, the terminal updates as follows.
                    \begin{equation}
                        \theta_{t+1} = \begin{cases}
                            [M_t,j_t,\Xi_{t+1}(c_t)]   & c_t = 0  \\
                            [M_t,j_t,\Xi_{t+1}(c_t)]   & c_t = 1  \\
                            [M_t+1,j_t,\Xi_{t+1}(c_t)] & c_t = e.
                        \end{cases}
                    \end{equation}
          \end{itemize}
    \item Increment $t$ by $1$ and go to step $2$.
\end{enumerate}
The evolution of the system under the reservation protocol is illustrated in Fig.~\ref{fig:ReservationProtocol}.
\begin{figure}[t]
    \centering
    \includegraphics[width=\columnwidth]{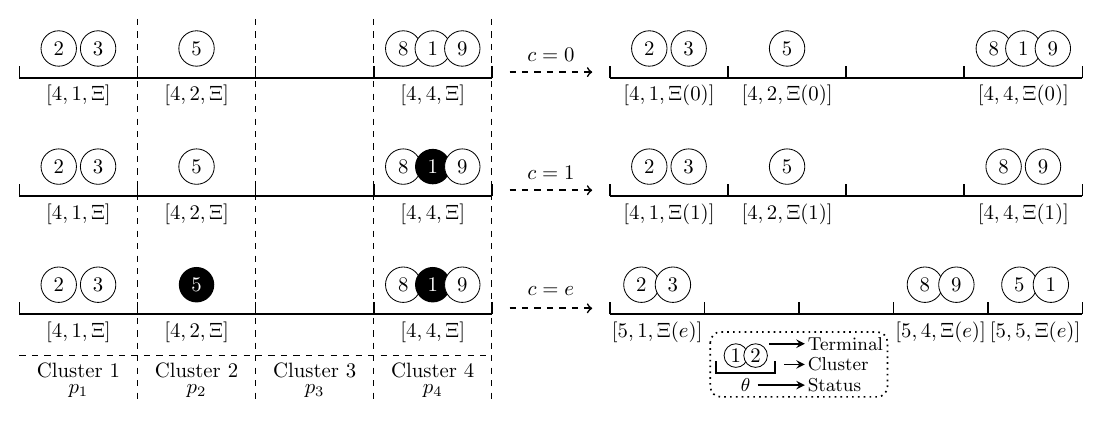}
    \caption{Illustration of the system evolution under the reservation protocol. Terminals that transmitted the reservation packets are highlighted in black, and the indices are assigned randomly for illustrative purposes only.}
    \label{fig:ReservationProtocol}
    \vspace{-1.5em}
\end{figure}
It is worth emphasizing that the same reservation protocol operates independently at each terminal and the operation relies solely on the information available to the individual terminal, specifically the channel feedback and the terminal's decision history.

\section{A POMDP Framework for Tree-Splitting-Based Reservation}
In this paper, we aim to find the optimal reservation protocol that adheres to the characteristics outlined in Section~\ref{sec:ReservationProtocol}, with the objective of minimizing the expected number of time slots required for channel reservation.

\subsection{A Goal-POMDP Formulation}\label{sec:POMDP}
We first notice that terminals can determine $M_t$ and $j_t$ in $\theta_t$ at any given time. Meanwhile, since no priority is assigned among terminals, it is sufficient for $\Xi_t$ to contain the distribution of active terminals across clusters. However, the terminals can only estimate this information. To account for the resulting uncertainty, we introduce the belief state $b$, which represents a probability distribution over all possible outcomes when active terminals are distributed across clusters. Specifically, let $G^N_M$ represent the set of all possible outcomes when $N$ identical terminals are distributed across $M$ different clusters. The belief state is then defined as a probability distribution over the set $\bigcup_{N,M}G^N_M$, where $N\geq0$ and $M\geq1$. Using the belief state, we can formulate the problem as a goal partially observable Markov decision process (Goal-POMDP). Specifically, the Goal-POMDP is formulated as follows.
\begin{itemize}
    \item State space $\mathcal{S}$: The state captures the distribution of the active terminals across clusters. Hence,
          \begin{equation}
              \mathcal{S} = \bigcup_{N,M}G^N_M,
          \end{equation}
          where $N\geq1$, and $M\geq1$.
    \item Target state space $\mathcal{T}$: The target state signifies the end of the reservation process, or equivalently, the point at which no active terminals remain. Hence,
          \begin{equation}
              \mathcal{T} = \bigcup_{M\geq1}G^0_M.
          \end{equation}
          Note that $\mathcal{S}$ and $\mathcal{T}$ do not contain $M_t$ and $j_t$ since these variables exist solely to guide the terminals in executing actions and responding to feedbacks. Hence, these two variables do not affect the optimization. The state is denoted by $s$. At state $s$, we also denote $s_N$ and $s_M$ as the number of active terminals and clusters, respectively.
    \item Action space $\mathcal{A}$: In each time slot, the reservation protocol specifies the transmission probability for each cluster. Hence, the feasible action depends on the state and, at state $s$, the action is denoted by
          \begin{equation}
              \bm{p}^s\in\mathcal{A}^s=\left\{\bm{p}\mid \bm{p}\in[0,1]^{s_M}\right\},
          \end{equation}
          where $[0,1]^{s_M}$ represents the set of vectors of size $s_M$, with each element taking a value within the range $[0,1]$. Then, the action space is $\mathcal{A} = \bigcup_{s\in\mathcal{S}}\mathcal{A}^s$.
    \item State transition probabilities $\mathcal{P}$: The probability of transitioning from state $s$ to state $s'$ upon applying action $\bm{p}^s$ is denoted by $P(s'\mid s,\bm{p}^s)$. Let $\eta_s^i$ be the number of terminals in the $i$-th cluster of state $s$. Then, we have
          \begin{equation}
              P(s'\mid s,\bm{p}^s) = \prod_{i=1}^{s_M}{\eta_s^i\choose \eta_{s'}^i}(1-p_i)^{\eta^i_{s'}}p_i^{\eta_s^i - \eta_{s'}^i},
          \end{equation}
          when any of the following three conditions hold.
          \begin{equation}\label{eq:TransitionProbabilityConditions}
              \begin{dcases}
                  \mathbbm{1}\left\{s'_N = s_N, s'_M = s_M+1, \sum_{i=1}^{s_M}(\eta_s^i - \eta_{s'}^i)>1\right\}=1 \\
                  \mathbbm{1}\left\{s'_N = s_N-1, s'_M =s_M, \sum_{i=1}^{s_M}(\eta_s^i - \eta_{s'}^i)=1\right\}=1  \\
                  \mathbbm{1}\{s'=s\}=1,
              \end{dcases}
          \end{equation}
          where $\mathbbm{1}\{A\}$ is the indicator function, which equals $1$ if $A$ is true and $0$ otherwise. If none of the above conditions hold, then $P(s'\mid s,\bm{p}^s)=0$.
    \item Observation space $\mathcal{O}$: The observation space consists of channel responses. Hence, $o\in\mathcal{O} = \{0,1,e\}$.
    \item Observation probabilities $\mathcal{Q}$: The probability of observing $o$ upon transitioning from state $s$ to state $s'$ after applying action $\bm{p}^s$ is denoted by $Q(o\mid s,s',\bm{p}^s)$ and satisfies
          \begin{equation}
              Q(0\mid s,s',\bm{p}^s) = \mathbbm{1}\{s'=s\},
          \end{equation}
          \begin{equation}
              Q(1\mid s,s',\bm{p}^s) = \mathbbm{1}\{s_N' = s_N-1, s_M' = s_M\},
          \end{equation}
          \begin{equation}
              Q(e\mid s,s',\bm{p}^s) = \mathbbm{1}\{s_N' = s_N, s_M' = s_M+1\}.
          \end{equation}
          Notably, the probability is independent of the action $\bm{p}^s$. Hence, for the remainder of this paper, we rewrite $Q(o\mid s,s',\bm{p}^s)$ as $Q(o\mid s,s')$.
    \item Cost $\mathcal{C}$: Since the objective is to minimize the expected time slots required to reach the target state, we have
          \begin{equation}
              C(s) = \begin{cases}
                  1 & s\in\mathcal{S}  \\
                  0 & s\in\mathcal{T}.
              \end{cases}
          \end{equation}
    \item Initial belief state $b_0$: At the beginning of each frame, all active terminals belong be a single cluster. Hence, $b_0$ is a probability distribution over the set $ \bigcup_{N\geq0}G^N_1$.
\end{itemize}
It is worth emphasizing that under the Goal-POMDP formulating, all active terminals can transmit the reservation packet in each time slot with nonzero probability. Thus, the Goal-POMDP formulation generalizes the canonical algorithms that only allow a predefined subset of active terminals to transmit the reservation packet in each time slot. Moreover, once the optimal policy for the Goal-POMDP is obtained, it can be applied independently by each terminal, relying solely on the channel feedback, which is broadcast to all terminals simultaneously. Hence, all terminals can seamlessly cooperate in executing the optimal action when starting with the same initial belief state.

\subsection{The Equivalent Belief MDP}
Since the terminals cannot fully observe the system state, canonical algorithms for fully observable MDP, such as the value iteration algorithm (VIA)~\cite{russell2016artificial}, cannot be directly applied. To find the optimal policy for the Goal-POMDP, a common approach is to reformulate it as a fully observable belief MDP~\cite{kaelbling1998}, denoted by $\mathcal{M}^b$. In this formulation, the belief state $b$, which is a probability distribution over $\mathcal{S}\cup\mathcal{T}$, is treated as the state of $\mathcal{M}^b$. The resulting state space is denoted by $\mathcal{B}$. To formally define $\mathcal{M}^b$, we first establish that for any belief state $b$, all states $s$ with nonzero probability share the same number of clusters $s_M$. This property follows from two key observations. First, the evolution of $s_M$ is deterministic and depends solely on the channel feedback, as given in~\eqref{eq:clusterNumberEvolution}. Second, the initial belief state assigns nonzero probability only to states where $s_M = 1$. Since the number of clusters evolves deterministically and originates from a single cluster, it remains consistent across all states with nonzero probability in any belief state $b$. Formally, for any two states $s, s' \in \{s \mid b(s) > 0\}$, we have $s_M = s'_M$, where $b(s)$ denotes the probability of being in state $s$ at belief state $b$.

Given this property, the feasible action $\bm{p}^b$ at belief state $b\in\mathcal{B}$ is given by
\begin{equation}\label{eq:BeliefMDPAction}
    \bm{p}^b \in\mathcal{A}^b = \{\bm{p}\mid \bm{p}\in[0,1]^{b_M}\},
\end{equation}
where $b_M$ is the number of clusters in states with nonzero probability at belief state $b$. The evolution of the belief state depends on the chosen action and the resulting observation. Specifically, the belief state resulting from observing $o$ after applying action $\bm{p}^b$ at belief state $b$ is given by
\begin{equation}\label{eq:BeliefMDPTransition}
    b^o_{\bm{p}^b}(s) = \begin{dcases}
        \frac{b_{\bm{p}^b}(s,o)}{b_{\bm{p}^b}(o)} & b_{\bm{p}^b}(o)\neq0 \\
        0                                         & b_{\bm{p}^b}(o)=0,
    \end{dcases}
\end{equation}
where
\begin{equation}
    b_{\bm{p}^b}(s,o) = \sum_{s'\in\mathcal{S}}Q(o\mid s',s)P(s\mid s',\bm{p}^b)b(s'),
\end{equation}
is the probability of being in state $s$ and observing $o$ after applying action $\bm{p}^b$ at belief state $b$, and
\begin{equation}
    b_{\bm{p}^b}(o) = \sum_{s\in\mathcal{S}}b_{\bm{p}^b}(s,o),
\end{equation}
is the probability of observing $o$ after applying action $\bm{p}^b$ at belief state $b$. The target state of $\mathcal{M}^b$ is defined as the belief state $b$ such that $b(s)=0$ for any non-target state $s\in\mathcal{S}$. Formally, the target state space $\mathcal{T}_{bel} \triangleq \{b\mid b(s)=0, \forall s\in\mathcal{S}\}$. It is important to note that the system will only terminate upon reaching the target state, even if no active terminals remain beforehand. The cost at belief state $b$ is the average cost weighted by state probabilities. Specifically, the cost is given by
\begin{equation}\label{eq:BeliefMDPcost}
    C^b(b) = \sum_{s\in\mathcal{S}}C(s)b(s) = \begin{cases}
        1 & b\in\mathcal{B}        \\
        0 & b\in\mathcal{T}_{bel}.
    \end{cases}
\end{equation}
With these quantities, $\mathcal{M}^b$ can be formulated as follows.
\begin{itemize}
    \item The state space is $\mathcal{B}$.
    \item The target state space is $\mathcal{T}_{bel}$.
    \item The action space $\mathcal{A}_{bel}=\bigcup_{b\in\mathcal{B}}\mathcal{A}^b$.
    \item The transition probabilities between belief states are characterized by~\eqref{eq:BeliefMDPTransition}.
    \item The cost is given by~\eqref{eq:BeliefMDPcost}.
\end{itemize}
It is noteworthy that a belief MDP is essentially an MDP in which the states are probability distributions. As a result, VIA can theoretically be applied to solve the belief MDP. Let $V(b)$ denote the value function of belief state $b$. The Bellman equation for the belief MDP is given by
\begin{equation}
    V(b) = \min_{\bm{p}^b\in\mathcal{A}^b}\left\{C^b(b) + \sum_{o\in\mathcal{O}}b_{\bm{p}^b}(o)V(b_{\bm{p}^b}^o)\right\},
\end{equation}
for all $b\in\mathcal{B}$. Given the Bellman equation, the optimal action can be computed numerically using VIA. However, for certain belief states, the problem's inherent properties allow us to determine the optimal action analytically. We formalize this in the following theorem.
\begin{theorem}\label{thm:OptimalActionBeliefMDP}
    Let $\mathcal{S}_1$ denote the set of state $s$ where $s_N = 1$, and let $\mathcal{B}_1$ denote the set of belief state with $b(s)=0$ for $s\in\mathcal{S}\setminus\mathcal{S}_1$. The policy $\pi$, which selects all clusters to transmit a reservation packet with probability one, is optimal for any belief state $b\in\mathcal{B}_1$. Formally, the optimal action at $b\in\mathcal{B}_1$ is
    \begin{equation}
        \bm{p}^{b,*}_i = \bm{p}^{b,\pi}_i = 1,
    \end{equation}
    for $1\leq i\leq b_M$ and the value function $V(b)=1$.
\end{theorem}
\begin{proof}
    First, we show that belief state $b\in\mathcal{B}_1$ is reachable from the initial belief state $b_0$. Recall that $b_0$ is a probability distribution over the state $s$ with $s_M=1$ and $s_N\leq N_{max}$, where we assume the number of active terminals is finite and upper bounded by $N_{max}$. Now, consider the case where the channel feedback $c=1$ is received, indicating that one terminal has won the competition and become inactive. Consequently, the resulting belief state $b'$ satisfies $b(s)\geq0$ only for state $s$ with $s_N\leq N_{max}-1$. If the actual number of active terminals at the beginning of the reservation process is at least $N_{max}-1$, the channel feedback $c=1$ will be received at least $N_{max}-1$ times. After $N_{max}-1$ such events, the belief state $b$ will satisfy $b(s)\geq0$ only for the state $s$ with $s_N=1$. Since this scenario occurs with nonzero probability, we conclude that the belief state $b\in\mathcal{B}_1$ is reachable. Notably, the terminals do not know which cluster contains the remaining active terminal.

    In the sequel, we prove that the action specified in the theorem is optimal when the system reaches the belief state $b\in\mathcal{B}_1$. To this end, we leverage the iterative nature of VIA. The interim value function $V_i(b)$ at iteration $i$ of VIA is updated as follows.
    \begin{equation}\label{eq:BeliefStateVIAUpdate}
        V_{i+1}(b) = \min_{\bm{p}^b\in\mathcal{A}^b}\left\{C^b(b) + \sum_{o\in\mathcal{O}}b_{\bm{p}^b}(o)V_i(b_{\bm{p}^b}^o)\right\}.
    \end{equation}
    As $i\rightarrow\infty$, $V_i(b)$ converges to the value function $V(b)$~\cite{russell2016artificial}. Therefore, it suffices to show that the action specified in the theorem is optimal for $V_i(b)$ at any iteration $i$. To prove this, we first notice from~\eqref{eq:BeliefStateVIAUpdate} that $V_i(b)\geq0$ under non-negative initializations. Hence, we have
    \begin{equation}\label{eq:InterimValueFunctionLowerBound}
        V_{i+1}(b)\geq C^b(b).
    \end{equation}
    Now, we consider applying the policy $\bm{p}^{b,\pi}$ at belief state $b\in\mathcal{B}_1$. Since at most one active terminal exists at belief state $b\in\mathcal{B}_1$, we distinguish between the following two cases.
    \begin{itemize}
        \item When there is only one active terminal, applying $\bm{p}^{b,\pi}$ results in channel feedback $c=1$, leading to a belief state $b'$ where $b'(s)\geq0$ only for state $s$ with $s_N=0$, which corresponds to a target state.
        \item When there is no active terminal, applying $\bm{p}^{b,\pi}$ results in channel feedback $c=0$. Since all clusters were selected for transmission, terminals can correctly infer that no active terminal remains, which is also a target state.
    \end{itemize}
    Since the target state has value function equal to zero by definition, we have
    \begin{equation}\label{eq:InterimValueFunctionOptimalAction}
        V_{i+1}^{\bm{p}^{b,\pi}}(b) = C^b(b),
    \end{equation}
    where $V_{i+1}^{\bm{p}^{b,\pi}}(b)$ is the interim value function resulting from taking action $\bm{p}^{b,\pi}$ at belief state $b$ and iteration $i$. Combining~\eqref{eq:BeliefStateVIAUpdate},~\eqref{eq:InterimValueFunctionLowerBound}, and~\eqref{eq:InterimValueFunctionOptimalAction} yields
    \begin{equation}\label{eq:InterimValueFunctionResult}
        V_{i+1}(b) = V_{i+1}^{\bm{p}^{b,\pi}}(b) = C^b(b) = 1.
    \end{equation}
    Since~\eqref{eq:InterimValueFunctionResult} holds at any iteration $i$ and $\lim_{i\rightarrow\infty}V_i(b) = V(b)$, we can conclude the proof.
\end{proof}

Despite Theorem~\ref{thm:OptimalActionBeliefMDP}, numerically solving the belief MDP using VIA remains challenging due to the infinite and continuous nature of the state space. While discretizing the probability distribution can mitigate this issue, achieving sufficiently accurate approximations requires solving an MDP with a prohibitively large state space. To address this, for belief states $b\notin\mathcal{B}_1$, we adopt the RTDP with belief state (RTDP-Bel)~\cite{bonet2000}.

\subsection{RTDP-Bel: An RL-Empowered Solution}\label{sec:RTDP-Bel}
RTDP-Bel defines a function $\tilde{V}(b)$ for belief state $b$ and maintains a hash table $\mathcal{H}$ to store these values. To approximate $V(b)$ using $\tilde{V}(b)$, RTDP-Bel simulates the system for multiple trials, with each trial terminating upon reaching the target state. Along the trajectory of each trial, $\tilde{V}(b)$ is either added to the hash table if $b$ is encountered for the first time or updated if $b$ has been visited previously. The same $\mathcal{H}$ is used across the trials, and as the number of trials increases, $\tilde{V}(b)$ is continuously updated to approximate the value function $V(b)$. Specifically, at each step within a trial, assuming the system is at belief state $b$, $\tilde{V}(b)$ is updated by solving the following optimization problem.
\begin{equation}\label{eq:RTDP-BelUpdate}
    \tilde{V}(b) = \min_{\bm{p}^b\in\mathcal{A}^b}\left\{C^b(b) + \sum_{o\in\mathcal{O}}b_{\bm{p}^b}(o)\tilde{V}(b_{\bm{p}^b}^o)\right\},
\end{equation}
where $\tilde{V}(b_{\bm{p}^b}^o)$ is retrieved from $\mathcal{H}$ if $b_{\bm{p}^b}^o$ has been visited before and initialized otherwise. To advance the simulation, the optimal action from~\eqref{eq:RTDP-BelUpdate} is applied, and the resulting observation is obtained. Then, the resulting belief state is computed using~\eqref{eq:BeliefMDPTransition}. Unlike VIA, which exhaustively evaluates all states, RTDP-Bel focuses only on states encountered in the trials and progressively learns to optimize its performance throughout the execution.

\subsection{Quantization and Discretization for Efficient RTDP-Bel}
The computational complexity of a direct implementation of RTDP-Bel is substantial. In practice, we introduce the following modifications to improve the efficiency.

The continuous nature of the probability distribution can lead to an excessively large hash table $\mathcal{H}$. To control the growth of $\mathcal{H}$ and ensure consistent updates to $\tilde{V}(b)$, we quantize the belief state $b$. Specifically, each probability $b(s)$ is quantized into discrete values based on a quantization parameter $q\in\mathbbm{Z}^+$ before interacting with $\mathcal{H}$. The quantized probability $q(b(s))$ is the given by
\begin{equation}\label{eq:BeliefStateQuantization}
    q(b(s)) = \frac{\lfloor b(s)\cdot q\rceil}{q},
\end{equation}
where $\lfloor\cdot\rceil$ denotes the rounding operator. It is important to note that these quantized values do not necessarily form a valid probability distribution, as they may not sum to one. Moreover, quantization is applied only during interactions with $\mathcal{H}$. This quantization ensures that the size of $\mathcal{H}$ grows more gradually and that $\tilde{V}(q(b))$ is updated whenever the system reaches a belief state that quantizes to $q(b)$, allowing for more frequent updates. However, the choice of $q$ is critical. A large $q$ causes rapid hash table growth and insufficient updates to $\tilde{V}(q(b))$. Conversely, a small $q$ may treat distinct belief states as identical.

The decision variable $\bm{p}^b$ in~\eqref{eq:RTDP-BelUpdate} is continuous, which increases computational complexity, as the optimization problem must be solved repeatedly in RTDP-Bel. To mitigate this, we approximate the solution by discretizing the decision variable. Specifically, the transmission probability $\bm{p}^b_i$ for each cluster is restricted to the following discrete set.
\begin{equation}\label{eq:DecisionVariableDiscrete}
    \bm{p}^b_i \in\left\{\frac{k}{d},\quad k\in\{0,1,...,d\}\right\},
\end{equation}
where $d$ governs the trade-off between approximation accuracy and computational cost. A larger $d$ yields a more precise approximation but increases computational cost.

Combining the practical enhancements with previously established results, RTDP-Bel for a single reservation process is summarized in Algorithm~\ref{Alg:RTDP-BelSingleTrial}.
\begin{algorithm}[!t]
    \begin{algorithmic}[1]
        \Procedure{RTDP-Bel}{$\mathcal{H},b_0$}
        \State Initialize $b = b_0$.
        \While{$b(s)>0$ for any $s$ with $s_N>0$}
        \If{$b\in\mathcal{B}_1$}
        \State $(\bm{p}^{b,*},\tilde{V}(b))\leftarrow$ obtain using Theorem~\ref{thm:OptimalActionBeliefMDP}.
        \Else
        \State $\{b^o_{\bm{p}^b}\}\leftarrow$ obtain using~\eqref{eq:BeliefMDPTransition} for each $\bm{p}^b$.
        \For{$b'\in\{b^o_{\bm{p}^b}\}$}
        \State $q(b')\leftarrow$ quantize $b'$ using~\eqref{eq:BeliefStateQuantization}.
        \If{$q(b')\in\mathcal{H}$}
        \State $\tilde{V}(q(b'))\leftarrow$ retrieve from $\mathcal{H}$.
        \Else
        \State $\tilde{V}(q(b'))\leftarrow$ initialization.
        \EndIf
        \EndFor
        \State $(\bm{p}^{b,*},\tilde{V}(b))\leftarrow$ solution to~\eqref{eq:RTDP-BelUpdate} and~\eqref{eq:DecisionVariableDiscrete}.
        \EndIf
        \State $q(b)\leftarrow$ quantize $b$ using~\eqref{eq:BeliefStateQuantization}.
        \If{$q(b)\in\mathcal{H}$}
        \State $\tilde{V}(q(b))\leftarrow$ update using $\tilde{V}(b)$.
        \Else
        \State Create entry $(q(b),\tilde{V}(b))$ in $\mathcal{H}$.
        \EndIf
        \State Apply the action $\bm{p}^{b,*}$.
        \State Obtain the resulting observation $o^*$.
        \State $b=b_{\bm{p}^{b,*}}^{o^*}$.
        \EndWhile
        \State \Return $\mathcal{H}$ and $\bm{p}^{b,*}$ for $q(b)\in\mathcal{H}$.
        \EndProcedure
    \end{algorithmic}
    \caption{RTDP-Bel for A Single Reservation Process}
    \label{Alg:RTDP-BelSingleTrial}
\end{algorithm}
The algorithm is executed for each reservation process, and as the number of executions increases, $\tilde{V}(b)$ progressively converges to $V(b)$.

\subsection{Hidden and Exposed Nodes}
Under the proposed reservation protocol, the hidden and exposed node problem, commonly encountered in wireless networks, can be alleviated similarly to the IEEE 802.11 RTS/CTS mechanism. Specifically, terminals first transmit an RTS packet before transmitting data packets, regardless of potential interference from other terminals. Packet transmission proceeds if no interference is detected, as indicated by the successful reception of a CTS packet. In the presence of interference, terminals defer transmission until the ongoing transmission is complete, which is facilitated by the retransmission of RTS packets. Afterward, the interfering terminals initiate an additional channel reservation using the proposed reservation protocol to eliminate further interference.

\section{Genie-Aided Pre-Training}\label{sec:ReservationProtocolImplemantation}
This section presents a pre-training strategy for RTDP-Bel based on an MDP formulated with genie assistance. Specifically, the value function of the genie-aided MDP provides an informative initialization for RTDP-Bel. While the corresponding optimal policy cannot be implemented without genie assistance, it provides a benchmark for evaluating the proposed reservation protocol. We begin by formulating the genie-aided MDP.

\subsection{The Formulation of Genie-Aided MDP}~\label{sec:MDP}
We consider a genie-aided problem where a genie provides terminals with the distribution of active terminals across clusters. With this additional information, the problem can be formulated as a Goal-MDP, characterized by $\mathcal{M} = (\mathcal{S},\mathcal{T},\mathcal{A},\mathcal{P},\mathcal{C})$, where each component is as specified in Section~\ref{sec:POMDP}. Then, the value function of $\mathcal{M}$ serves as the informative initialization for $\tilde{V}(b)$ and its optimal policy provides a benchmark for evaluating the optimality of the proposed reservation protocol.

We first establish that the optimal policy for $\mathcal{M}$ is well-defined. According to~\cite{bertsekas1995}, it suffices to show that the target state is reachable from every state. To this end, we consider a policy under which all active terminals transmit a reservation packet with a unified probability $p$ in each time slot. Then, for state $s$ with $s_N<\infty$, the expected number of time slots $\mathcal{E}$ required to reach the target state is given by
\begin{equation}
    \mathcal{E} = \sum_{k=1}^{s_N}\frac{1}{kp(1-p)^{k-1}}.
\end{equation}
Since the number of active terminals at each time slot is known in $\mathcal{M}$, we consider the policy that adopts $p = \frac{1}{k}$ in each time slot, where $k$ is the current number of active terminals. Then, we have
\begin{equation}
    \mathcal{E} = \sum_{k=1}^{s_N}\frac{1}{(1-\frac{1}{k})^{k-1}} \leq s_N\cdot e < \infty,
\end{equation}
where we use the inequality $(1-\frac{1}{k})^{k-1}\geq e^{-1}$. Thus, the target state is reachable from any state, ensuring that the optimal policy is well-defined. The value function, denoted by $V(s)$, satisfies the following Bellman equation.
\begin{equation}\label{eq:BellmanEquation}
    V(s) = \min_{\bm{p}^s\in\mathcal{A}^s}\left\{C(s) + \sum_{s'\in\mathcal{S}}P(s'\mid s,\bm{p}^s)V(s')\right\},
\end{equation}
where $V(s) = 0$ for $s\in\mathcal{T}$ by definition. To solve~\eqref{eq:BellmanEquation}, we adopt the canonical VIA, which iteratively updates interim value functions until convergence. Specifically, let $V_i(s)$ be the interim value function at iteration $i$ of VIA. Then, for all $s\in\mathcal{S}$, the update rule is given by
\begin{equation}\label{eq:VIA}
    V_{i+1}(s) = \min_{\bm{p}^s\in\mathcal{A}^s}\left\{C(s) + \sum_{s'\in\mathcal{S}}P(s'\mid s,\bm{p}^s)V_i(s')\right\},
\end{equation}
where $V_0(s) = 0$ for $s\in\mathcal{S}$ by initialization, and $V_i(s) = V(s) = 0$ for $s\in\mathcal{T}$ at any iteration $i$ by definition. When updated according to~\eqref{eq:VIA}, we have $\lim_{i\rightarrow\infty}V_i(s) = V(s)$ for $s\in\mathcal{S}$~\cite{russell2016artificial}. Upon convergence, the optimal policy can be obtained as
\begin{equation}\label{eq:OptimalPolicy}
    \pi^*(s) = \argmin_{\bm{p}^s\in\mathcal{A}^s}\left\{C(s) + \sum_{s'\in\mathcal{S}}P(s'\mid s,\bm{p}^s)V(s')\right\},
\end{equation}
for all $s\in\mathcal{S}$.

Each VIA iteration requires solving the optimization problem in~\eqref{eq:VIA} for every state $s\in\mathcal{S}$. However, the size of the state space, denoted by $|\mathcal{S}|$, can be theoretically infinite due to the unbounded number of clusters. While truncating the state space by limiting the number of clusters to $M_{max}$ reduces the size, the resulting state space remains large, with a size of
\begin{equation}
    |\mathcal{S}| = \sum_{N = 1}^{N_{max}}\sum_{m = 1}^{M_{max}}{N+m-1\choose m-1},
\end{equation}
where $N_{max}$ is the maximum number of active terminals. To mitigate the challenges posed by the large state space, we exploit the structural properties of $\mathcal{M}$ to formulate an equivalent MDP, which significantly reduces the computational complexity of solving the Bellman equation.

\subsection{A Low-Complexity Equivalent MDP}\label{sec:StructualProperty}
We start with introducing two structure properties of $\mathcal{M}$, which are essential for formulating the low-complexity equivalent MDP. To this end, we recall that clusters may become empty during the reservation process, and empty clusters do not affect the system evolution. This leads us to the first structure property. Specifically, we define an operator $\mathscr{R}$, where $\mathscr{R}(s)$ represents the state obtained by removing all clusters with zero terminals in state $s$. Then, we have the following lemma.
\begin{lemma}\label{lem:RemoveZeros}
    For $i\geq0$ and $s\in\mathcal{S}\cup\mathcal{T}$, $V_i(\mathscr{R}(s)) = V_i(s)$.
\end{lemma}
\begin{proof}
    For $s\in\mathcal{T}$, the result holds trivially since $V_i(s)=0$ for any $i\geq0$ by definition. Now, for $s\in\mathcal{S}$, we use mathematical induction  to establish the result. Specifically, the base case $i = 0$ holds by initialization. Assume that the lemma is true for iteration $i$. We now show that it holds at iteration $i+1$. We first notice that $C(s) = C(\mathscr{R}(s))$. Therefore, it suffices to prove that
    \begin{equation}\label{eq:lemma1}
        P(\mathscr{R}(s')\mid\mathscr{R}(s),\mathscr{R}(\bm{p}^s)) = P(s'\mid s,\bm{p}^s),
    \end{equation}
    where $\mathscr{R}(\bm{p}^s)$ is the vector obtained by removing all the $p_i$ in $\bm{p}^s$ corresponding to zero-terminal clusters. To this end, we first notice that $P(\mathscr{R}(s')\mid\mathscr{R}(s),\mathscr{R}(\bm{p}^s))=0$ if $P(s'\mid s,\bm{p}^s)=0$. Then,  for the case where $P(s'\mid s,\bm{p}^s)\neq0$, we define $g_1(s)$ and $g_2(s)$ as the sets of clusters with and without terminals at state $s$, respectively. Since the number of terminals in each cluster is non-negative and non-increasing, we have
    \begin{equation}
        \prod_{i\in g_2(s)}{\eta_s^i\choose \eta_{s'}^i}(1-p_i)^{\eta^i_{s'}}p_i^{\eta_s^i - \eta_{s'}^i} = 1.
    \end{equation}
    Hence, we have
    \begin{equation}
        \begin{split}
            P(\mathscr{R}(s')\mid\mathscr{R}(s),\mathscr{R}(\bm{p}^s)) = & \prod_{i\in g_1(s)}{\eta_s^i\choose \eta_{s'}^i}(1-p_i)^{\eta^i_{s'}}p_i^{\eta_s^i - \eta_{s'}^i} \\
            =                                                            & P(s'\mid s,\bm{p}^s).
        \end{split}
    \end{equation}
    Combining this with the induction hypothesis at iteration $i$, we have
    \begin{equation}
        \begin{split}
            V_{i+1}(\mathscr{R}(s))  = & \min_{\bm{p}^s\in\mathcal{A}^s}\left\{C(s) + \sum_{s'\in\mathcal{S}}P(s'\mid s,\bm{p}^s)V_i(\mathscr{R}(s'))\right\} \\
            =                          & V_{i+1}(s),
        \end{split}
    \end{equation}
    which completes the proof.
\end{proof}
By Lemma~\ref{lem:RemoveZeros}, the VIA update of states with empty clusters can directly replicate the VIA update of the state obtained by removing zero-terminal clusters. Let $\mathcal{N}_1$ be the number of optimization problems that need to be solved per VIA iteration after applying Lemma~\ref{lem:RemoveZeros}. Then, we have
\begin{equation}\label{eq:UpdateReduction1}
    \mathcal{N}_1 = \sum_{N=1}^{N_{max}}\sum_{m=1}^N{N-1\choose m-1}.
\end{equation}
Since the interim value functions at every iteration of the VIA satisfy the property in Lemma~\ref{lem:RemoveZeros}, we can conclude that the value function of $\mathcal{M}$ possesses the same property.

Another structural property arises from the observation that the cluster order does not affect system evolution. Specifically, we define the operator $\mathscr{S}$, where $\mathscr{S}(s)$ represents the state obtained by swapping the position of any two clusters in state $s$. Then, the following lemma holds.
\begin{lemma}\label{lem:SwapPosition}
    For $i\geq0$ and $s\in\mathcal{S}\cup\mathcal{T}$, $V_i(\mathscr{S}(s)) = V_i(s)$.
\end{lemma}
\begin{proof}
    The proof follows a similar structure to the proof of Lemma~\ref{lem:RemoveZeros}. For $s\in\mathcal{T}$, the lemma holds since $V_i(s)=0$ for any $i\geq0$ by definition. For $s\in\mathcal{S}$, we use mathematical induction to prove the result. Since swapping clusters does not affect summation over clusters, the cost function remains invariant, i.e., $C(s) = C(\mathscr{S}(s))$. Similarly, the three conditions in~\eqref{eq:TransitionProbabilityConditions} remain unchanged under $\mathscr{S}$. Furthermore, products over clusters also remain invariant after swapping clusters, implying $P(\mathscr{S}(s')\mid\mathscr{S}(s),\mathscr{S}(\bm{p}^s)) =  P(s'\mid s,\bm{p}^s)$, where $\mathscr{S}(\bm{p}^s)$ is the vector obtained by swapping the elements corresponding to the clusters swapped by $\mathscr{S}$. Combining these with the induction hypothesis at iteration $i$, we have
    \begin{equation}
        \begin{split}
            V_{i+1}(\mathscr{S}(s)) = & \min_{\bm{p}^s\in\mathcal{A}^s}\left\{C(s) + \sum_{s'\in\mathcal{S}}P(s'\mid s,\bm{p}^s)V_i(\mathscr{S}(s'))\right\} \\
            =                         & V_{i+1}(s),
        \end{split}
    \end{equation}
    which concludes the proof.
\end{proof}
By Lemma~\ref{lem:SwapPosition}, the number of optimization problems that need to be solved per VIA iteration can be further reduced to
\begin{equation}\label{eq:UpdateReduction2}
    \mathcal{N}_2 = \sum_{N=1}^{N_{max}}\sum_{m=1}^Np(N,m),
\end{equation}
where the partition function $p(N,m)$ satisfies the recursive relation $p(N,m) = p(N-m,m)+p(N-1,m-1)$, with base cases $p(N,1) = p(N,N) = 1$ if $N\geq1$ and $p(N,m)=p(N,0)=p(0,m)=0$ if $m>N>0$. Since the interim value functions at every VIA iteration satisfy the property in Lemma~\ref{lem:SwapPosition}, the value function of $\mathcal{M}$ possesses the same property.

Building on the key structural properties of the value functions, we introduce an equivalent Goal-MDP, denoted by $\mathcal{M}^-$. The optimal policy for $\mathcal{M}$ can be directly derived from that for $\mathcal{M}^-$, while solving $\mathcal{M}^-$ is significantly less computationally demanding. To formulate $\mathcal{M}^-$, we first introduce $\tilde{G}^N_M$, where $M\leq N$, as the set of all possible distributions of $N$ identical terminals across $M$ identical clusters, with each cluster contains at least one terminal. Without loss of generality, the clusters are also ordered in ascending order based on the number of terminals they contain. Then, $\mathcal{M}^- = (\mathcal{S}^-,\mathcal{T}^-,\mathcal{A}^-,\mathcal{P}^-,\mathcal{C}^-)$ is defined as follows.
\begin{itemize}
    \item The state space $\mathcal{S}^-$ is analogous to that in $\mathcal{M}$, with $G^N_M$ replaced by $\tilde{G}^N_M$. It is important to note that $\mathcal{S}^-\subset\mathcal{S}$.
    \item The target state space $\mathcal{T}^- = \{s\mid s_N=0\}$.
    \item The feasible action at state $s\in\mathcal{S}^-$ satisfies $\bm{p}^{s} \in \mathcal{A}^{s}= \left\{\bm{p}\mid\bm{p}\in[0,1]^{s_M}\right\}$.
    \item The transition probability from $s\in\mathcal{S}^-$ to state $s'\in\mathcal{S}^-$ under action $\bm{p}^{s}$ is given by
          \begin{equation}\label{eq:ReducedStateTransition}
              P^-(s'\mid s,\bm{p}^s) = \sum_{\tilde{s}\in\mathcal{F}(s')}P(\tilde{s}\mid s,\bm{p}^s),
          \end{equation}
          where $P(\tilde{s}\mid s,\bm{p}^s)$ is as defined in Section~\ref{sec:POMDP}, and $\mathcal{F}(s') = \{\tilde{s}\mid \mathscr{F}(\tilde{s}) = s', \tilde{s}\in\mathcal{S}\}$. The operator $\mathscr{F}$ removes clusters with zero terminals and arranges the remaining clusters in ascending order according to the number of terminals.
    \item The cost function is defined as
          \begin{equation}
              C^-(s) = \begin{cases}
                  1 & s\in\mathcal{S}^-  \\
                  0 & s\in\mathcal{T}^-.
              \end{cases}
          \end{equation}
\end{itemize}
Let $V^-(s)$ and $\pi^{-,*}$ be the value function and optimal policy for $\mathcal{M}^-$, respectively. Then, we can establish a direct correspondence between $\mathcal{M}^-$ and $\mathcal{M}$.
\begin{theorem}\label{thm:OptimalPolicyReducedMDP}
    The value function $V(s)$ and optimal policy $\pi^*$ for $\mathcal{M}$ can be obtained by mapping states through the operator $\mathscr{F}$. Specifically, for each $s\in\mathcal{S}$, we have
    \begin{equation}\label{eq:ValueFunctionTransfer}
        V(s) = V^-(\mathscr{F}(s)),
    \end{equation}
    and the optimal policy $\pi^*$ is
    \begin{equation}\label{eq:OptimalPolicyTransfer}
        \pi^*(s) = \pi^{-,*}(\mathscr{F}(s)).
    \end{equation}
\end{theorem}
\begin{proof}
    By definition, the operator $\mathscr{F}$ maps states in $\mathcal{S}$ to their corresponding states in $\mathcal{S}^-$ by removing empty clusters and reordering. Thus, equation~\eqref{eq:ValueFunctionTransfer} follows directly from Lemma~\ref{lem:RemoveZeros} and Lemma~\ref{lem:SwapPosition}. Since the optimal policy is derived from the value function, equation~\eqref{eq:OptimalPolicyTransfer} follows immediately.
\end{proof}
As indicated by~\eqref{eq:UpdateReduction1} and~\eqref{eq:UpdateReduction2}, VIA for $\mathcal{M}^-$ requires solving significantly fewer optimization problems per iteration than for $\mathcal{M}$. With the established results, we can now introduce the pre-training strategy for RTDP-Bel using the value function of $\mathcal{M}$.

\subsection{Genie-Aided Value Iteration for Pre-Training}
The value function $V(s)$ of $\mathcal{M}$ can be viewed as a pre-trained $\tilde{V}(b)$ and serves as an informative initialization to $\tilde{V}(b)$. Specifically, when a belief state $b$ is visited for the first time, $\tilde{V}(b)$ is initialized as
\begin{equation}\label{eq:BelifStateValueInitialization}
    \tilde{V}(b) = \sum_{s\in\mathcal{S}}b(s)V(s),
\end{equation}
As demonstrated in~\cite{littman1995}, informative initialization provides good performance after only a few trials, in contrast to zero initialization. According to Theorem~\ref{thm:OptimalPolicyReducedMDP}, obtaining $V(s)$ reduces to computing $V^-(s)$, which can be calculated using VIA with a modified update rule given by
\begin{equation}\label{eq:CompactOptimization1}
    V^-_{i+1}(s) = \min_{\bm{p}^s\in\mathcal{A}^s}\left\{C^-(s) + \sum_{s'\in\mathcal{S}^-}P^-(s'\mid s,\bm{p}^s)V^-_i(s')\right\},
\end{equation}
where $V^-_i(s)$ is the interim value function of $\mathcal{M}^-$ at $i$-th VIA iteration and $V^-_i(s)=0$ for $s\in\mathcal{T}^-$ and all $i$. The convergence of the VIA in this paper is defined as the point where the maximum difference between consecutive iterations is less than a predetermined threshold $\epsilon$. Upon convergence, the optimal policy $\pi^{-,*}$ is given by
\begin{equation}\label{eq:OptimalPolicyFromVF}
    \pi^{-,*}(s) = \argmin_{\bm{p}^s\in\mathcal{A}^s}\left\{C^-(s) + \sum_{s'\in\mathcal{S}^-}P^-(s'\mid s,\bm{p})V^-(s')\right\},
\end{equation}
for $s\in\mathcal{S}^-$. Then, the value function and optimal policy for $\mathcal{M}$ can be obtained by applying Theorem~\ref{thm:OptimalPolicyReducedMDP}.

To further reduce the computational complexity, we observe that in certain states, the value function and optimal policy for $\mathcal{M}^-$ can be determined analytically. To this end, we define $N_c(s)$ as the set of clusters containing exactly one terminal in state $s$ and introduce $\mathcal{G}$ as the set of state $s\in\mathcal{S}^-$ satisfying $|N_c(s)|\geq s_N$. Then, we have the following theorem.
\begin{theorem}\label{lem:SpecialState}
    For any state $s\in\mathcal{G}$, the policy $\pi$, which selects a single cluster $k\in N_c(s)$ to transmit a reservation packet with probability one while assigning zero probability to all other clusters, is optimal. Formally, the optimal action at $s\in\mathcal{G}$ is given by
    \begin{equation}
        \bm{p}^{s,*}_i = \bm{p}^{s,\pi}_i = \mathbbm{1}\{i=k\},
    \end{equation}
    and the value function is
    \begin{equation}
        V^-(s) = V^{-,\pi}(s) = s_N,
    \end{equation}
    where $V^{-,\pi}(s)$ is the value function resulting from applying policy $\pi$ at state $s$.
\end{theorem}
\begin{proof}
    According to the state transition probabilities detailed in~\eqref{eq:ReducedStateTransition}, the maximum step size by which $s_N$ can decrease in a single transition is one. Since the target state satisfies $s_N=0$, and $V^-(s)$ represents the minimum expected number of transitions required to reach the target state from $s$, it follows that
    \begin{equation}\label{eq:theorem2Result1}
        V^-(s)\geq s_N,
    \end{equation}
    for any $s\in\mathcal{S}^-$. Therefore, to establish the optimality of policy $\pi$, it suffices to demonstrate that the value function resulting from applying $\pi$ at state $s$ satisfies $V^{-,\pi}(s)=s_N$. To this end, we first notice that $s_N$ and $|N_c(s)|$ both decrease by one after the application of policy $\pi$ at state $s\in\mathcal{G}$. Hence, the condition $|N_c(s')|\geq s'_N$ holds at all intermediate state $s'$ along the path starting from state $s\in\mathcal{G}$ under policy $\pi$. Moreover, the state transitions along the path is deterministic since the action under policy $\pi$ is deterministic.

    With these in mind, we prove the theorem by induction, starting with the base case where $|N_c(s)|\geq s_N=1$. Under policy $\pi$, the system reaches the target state after a single state transition. Hence, we have
    \begin{equation}
        V^{-,\pi}(s) = 1 = s_N,
    \end{equation}
    for any state $s\in\mathcal{G}$ such that $|N_c(s)|\geq s_N=1$. Assume as the induction hypothesis that $V^{-,\pi}(s) = s_N=a$ holds for any state $s\in\mathcal{G}$ such that $|N_c(s)|\geq s_N=a$. We now prove that the hypothesis holds for $s\in\mathcal{G}$ with $|N_c(s)|\geq s_N=a+1$. To this end, we recall that, under policy $\pi$, the path starting from $s\in\mathcal{G}$ is deterministic. Hence, we have
    \begin{equation}
        V^{-,\pi}(s) = 1 + V^{-,\pi}(s'),
    \end{equation}
    where $s'$ satisfies $|N_c(s')| = |N_c(s)| - 1 $ and $s'_N = s_N - 1=a$. Then, we know $s'\in\mathcal{G}$ and $V^{-,\pi}(s') = a$ by the induction hypothesis. Hence, we have
    \begin{equation}
        V^{-,\pi}(s) =  1 + a = s_N
    \end{equation}
    where $s\in\mathcal{G}$ and satisfies $|N_c(s)|\geq s_N=a+1$. Then, we can conclude that the induction hypothesis holds for all $a$, which establishes the optimality of policy $\pi$ with $V^-(s) = V^{-,\pi}(s) = s_N$, completing the proof.
\end{proof}
\begin{algorithm}[!t]
    \begin{algorithmic}[1]
        \Procedure{Modified VIA}{$\mathcal{M}^-, \mathcal{M},\epsilon$}
        \State Initialize $i = 0$ and $V^-_0(s)=0$ for $s\in\mathcal{S}^-$.
        \State Set $V^-(s) = s_N$ for $s\in\mathcal{G}$.
        \Repeat
        \For{$s\in\mathcal{S}^-\setminus\mathcal{G}$}
        \State $V^-_{i+1}(s)\leftarrow$ solution to~\eqref{eq:CompactOptimization1} and~\eqref{eq:CompactOptimization2}.
        \EndFor
        \State $i = i + 1$.
        \Until{$\max_s\{|V^-_i(s) - V^-_{i-1}(s)|\}\leq\epsilon$}
        \For{$s\in\mathcal{G}$}
        \State $\pi^{-,*}(s)\leftarrow$ according to Theorem~\ref{lem:SpecialState}.
        \EndFor
        \For{$s\in\mathcal{S}^-\setminus\mathcal{G}$}
        \State $\pi^{-,*}(s)\leftarrow$ using~\eqref{eq:OptimalPolicyFromVF}.
        \EndFor
        \For{$s\in\mathcal{S}$}
        \State $(V(s),\pi^*(s))\leftarrow$ according to Theorem~\ref{thm:OptimalPolicyReducedMDP}.
        \EndFor
        \State \Return $\pi^*$ and $V(s)$ for $s\in\mathcal{S}$.
        \EndProcedure
    \end{algorithmic}
    \caption{Genie-Aided Value Iteration}
    \label{Alg:ModifiedVIA}
\end{algorithm}
\noindent With Theorem~\ref{lem:SpecialState}, when we apply the update in~\eqref{eq:CompactOptimization1} to obtain the value function $V^-(s)$ for $s\in\mathcal{S}^-$, the update for state $s\in\mathcal{G}$ can be avoided, and the interim value function satisfies
\begin{equation}\label{eq:CompactOptimization2}
    V^-_i(s) = s_N,
\end{equation}
for $s\in\mathcal{G}$ at any iteration $i$ of VIA. Combing the established results, the algorithm for computing the value function and the optimal policy for $\mathcal{M}$ is summarized in Algorithm~\ref{Alg:ModifiedVIA}. Once the value function $V(s)$ is obtained, the initialization of $\tilde{V}(b)$ follows~\eqref{eq:BelifStateValueInitialization}.

\section{Simulation Results}\label{sec:NumericalResults}
We consider a system where each channel reservation involves a maximum of $N_{max}=5$ terminals. Data packets arrive according to a Poisson process with arrival rate $\lambda$, and each packet is randomly assigned to one of the terminals with equal probability upon arrival. Time is divided into frames, each containing $T$ time slots. Let $\Lambda^T$ denote the number of packets arriving at one of the terminals within a frame of size $T$. Then, we have
\begin{equation}
    Pr(\Lambda^T = k) = \frac{(\lambda' T)^ke^{-\lambda' T}}{k!},
\end{equation}
where $\lambda' \triangleq \frac{\lambda}{N_{max}}$. Let $\Gamma^T$ represent the number of terminals receiving at least one packet during a frame of size $T$. Then, we can obtain
\begin{equation}
    Pr(\Gamma^T = k) = {N_{max}\choose k}\left(1-e^{-\lambda' T}\right)^ke^{-\lambda' T(N_{max}-k)},
\end{equation}
where $k\leq N_{max}$. The initial belief state $b_0$, as required in Algorithm~\ref{Alg:RTDP-BelSingleTrial}, is then given by $b_0(s) = P(\Gamma^T = k)$ where $s = G^{k}_{1}$. Throughout our simulations, the ACK/NACK feedback, given the number of simultaneously attempting terminals, is the same as that assumed in the IEEE 802.11 DCF and ALOHA. More specifically, an ACK is broadcast only when a single terminal attempts transmission, and a NACK is broadcast if more than one terminal attempts transmission. Then, the operation of the system under the two frame division strategies is detailed as follows.
\paragraph{Fixed frame} Channel reservation takes place at the start of each frame. If there is ongoing packet transmission at the beginning of the frame, it is paused to allow the channel reservation to proceed. Once the channel reservation concludes, packet transmission resumes. Note that all previously queued packets must be transmitted before the arrivals within the current frame. This is achieved by having the terminal send a finish signal once it has finished transmitting, and the corresponding channel feedback is broadcast to inform all terminals of the transmission status. By listening to the channel feedbacks, terminals can determine the number of pending transmissions in the queue and coordinate their transmission.

\paragraph{Dynamic frame} A new frame begins only after all terminals with backlogged data have completed their transmissions. This is made possible by leveraging the finish signal described earlier. This strategy will be the default setting for the subsequent simulations.

We simulate the system for $20000$ time slots. Each reservation packet, along with the corresponding feedback, has a fixed size of 30 bytes and carries no data payload. The finish signal and the corresponding feedback also have a fixed size of 30 bytes. The size of a data packet can vary. In the simulation, the default size is $180$ bytes, which gives a data-to-overhead ratio $\rho \triangleq \frac{180}{2\times30} = 3$. To maintain network stability, the data arrival rate is constrained by $\lambda\rho\leq 1$. For RTDP-Bel in Algorithm~\ref{Alg:RTDP-BelSingleTrial}, we take the following default values and settings.
\begin{itemize}
    \item The quantization parameter for $\mathcal{H}$ is $q = 10$.
    \item The discretization parameter in~\eqref{eq:DecisionVariableDiscrete} is $d=10$.
    \item The number of clusters is capped at a maximum of $15$. If a collision occurs and the number of clusters has reached this threshold, the colliding terminals remain in the same cluster rather than forming a new one.
    \item To reduce computational demand, no more than two clusters of terminals are allowed to transmit the reservation packet in each time slot.
\end{itemize}

The performance of the proposed reservation protocol is evaluated using two key metrics. The first metric is the effective throughput, denoted by $\gamma$. The effective throughput is defined as the average number of successfully transmitted data packets per time slot. The second metric is the average delay, denoted by $\tau$. The delay of a data packet is defined as the number of time slots elapsed from its generation to the completion of its transmission.

We start with comparing the proposed reservation protocol against the following benchmark protocols.
\begin{itemize}
    \item \textbf{Slotted ALOHA}: The Slotted ALOHA adopts a binary exponential backoff (BEB) mechanism with a maximum window size of $W_{max}=1024$. Specifically, after the $k$-th collision, the terminal selects a random backoff time within the range $[0,\min\{2^k,W_{max}\}-1]$.
    \item \textbf{Tree (Stack) algorithm}: The Tree algorithm follows the non-blocked stack algorithm $(N,N,U,2)$ defined in~\cite{tsybakov1985survey}. In this protocol, terminals retransmit immediately after a collision with probability $\frac{1}{2}$.
    \item \textbf{CSMA/CA}: The CSMA/CA adopts the RTS/CTS mechanism. A terminal first transmits an RTS packet before transmitting data packets and proceeds with transmission only if a CTS packet is received. The CTS packet is sent by the receiver only if a single terminal has sent an RTS packet. The RTS/CTS packets have a fixed size of 30 bytes. Additionally, CSMA/CA incorporates the BEB mechanism with $W_{max}=1024$. Specifically, after the $k$-th collision, the terminal picks a random backoff time within the range $[0,\min\{2^{k+2},W_{max}\}]$.
\end{itemize}

\begin{figure}[t]
    \centering
    \includegraphics[width=\columnwidth]{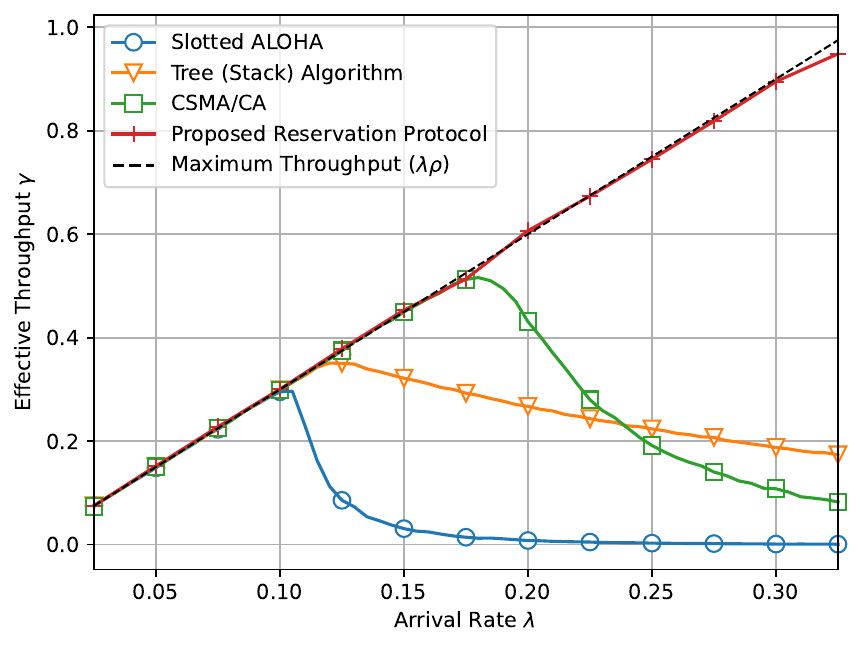}
    \caption{Comparison of the proposed reservation protocol with benchmark protocols. Some data points exceed the maximum effective throughput due to stochastic data arrival.}
    \label{fig:Benchmarks}
    \vspace{-1.5em}
\end{figure}

Fig.~\ref{fig:Benchmarks} illustrates the effective throughput achieved by both the benchmark protocols and the proposed reservation protocol with the dynamic frame strategy. As observed, the proposed reservation protocol outperforms the benchmark protocols, particularly under heavy traffic conditions. When traffic is light, all protocols efficiently deliver data packets, reaching the maximum effective throughput. However, as traffic intensity increases, the benchmark protocols suffer from excessive collisions and inefficient collision resolution mechanisms. Among the benchmark protocols, Slotted ALOHA exhibits the worst performance due to its simplistic retransmission strategy. The Tree (Stack) algorithm, which employs a more structured retransmission approach, and CSMA/CA, which utilizes RTS/CTS packets to manage large packet transmissions, perform comparatively better. Additionally, we notice that, for arrival rates in the range $\lambda\in[0.125,0.225]$, CSMA/CA outperforms the Tree (Stack) algorithm, but for $\lambda\geq0.225$, the trend reverses. This occurs because CSMA/CA relies on an ALOHA-like mechanism to manage RTS/CTS packet transmissions. As traffic increases, the inefficiencies associated with frequent RTS/CTS contention offset the benefits of the shorter RTS/CTS packets, leading to performance degradation. In contrast, the proposed reservation protocol, leveraging the POMDP framework and frame-based transmission scheme, achieves the best performance.

\begin{figure*}[!t]
    \centering
    \begin{subfigure}[b]{0.24\textwidth}
        \includegraphics[width=\textwidth]{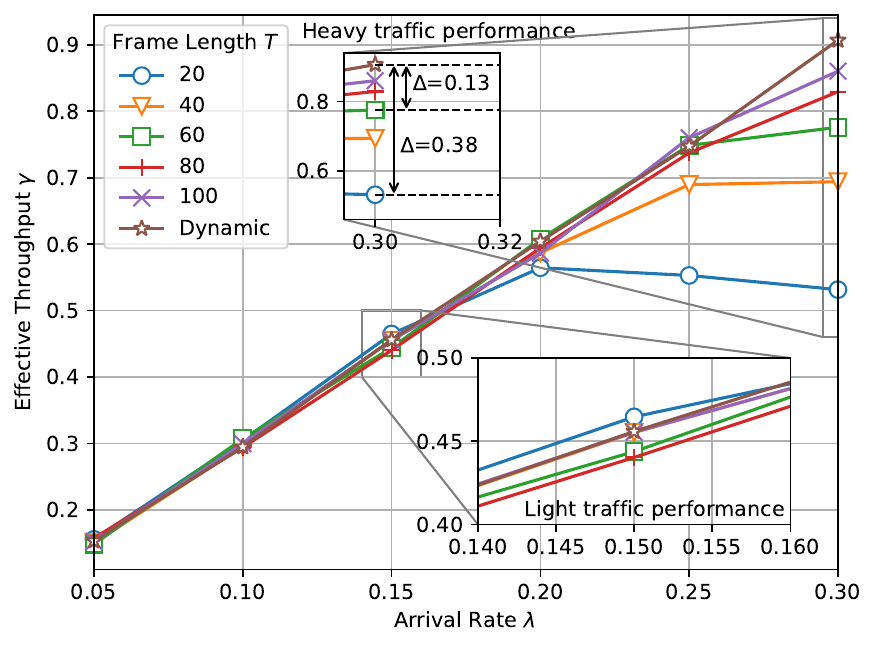}
        \vspace{-0.6cm}
        \caption{Effective throughput $\gamma$ under various frame division strategies.}
        \label{fig:FrameThroughput}
    \end{subfigure}
    \hfill
    \begin{subfigure}[b]{0.24\textwidth}
        \includegraphics[width=\textwidth]{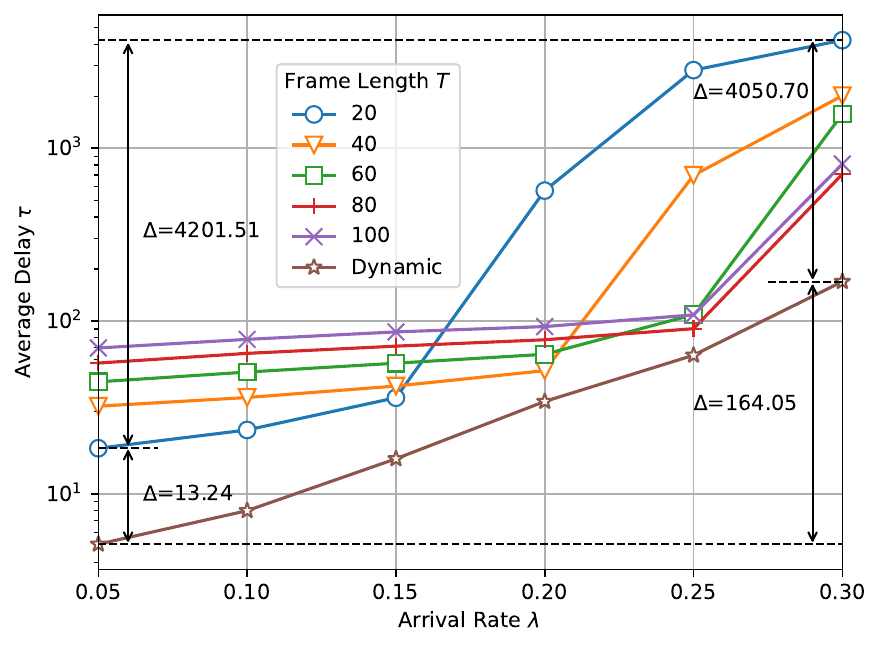}
        \vspace{-0.6cm}
        \caption{Average delay $\tau$ under various frame division strategies.}
        \label{fig:FrameDelay}
    \end{subfigure}
    \hfill
    \begin{subfigure}[b]{0.24\textwidth}
        \vspace{0.2cm}
        \includegraphics[width=\textwidth]{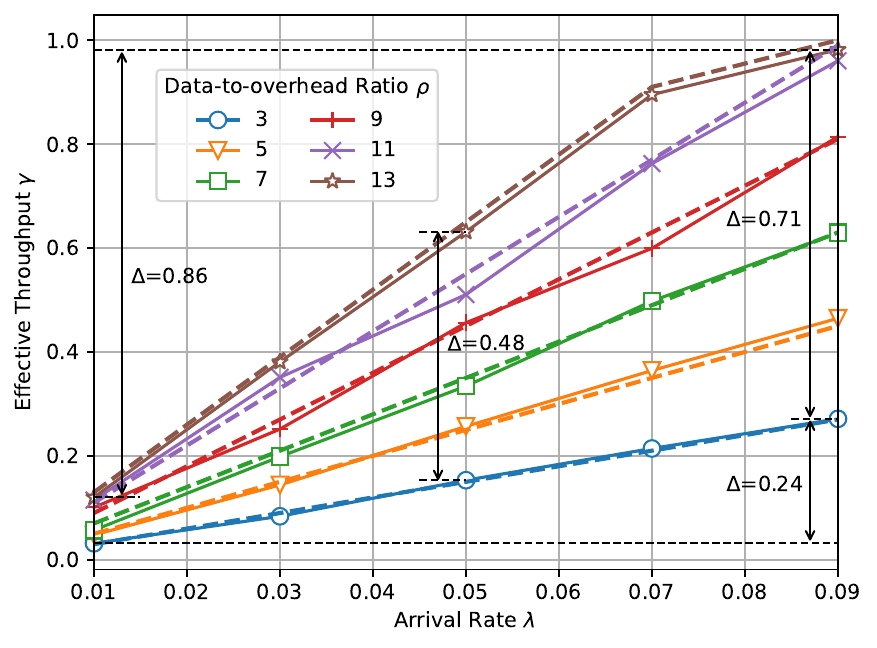}
        \vspace{-0.6cm}
        \caption{Effective throughput $\gamma$ when the data-to-overhead ratio varies.}
        \label{fig:RatioThroughput}
    \end{subfigure}
    \hfill
    \begin{subfigure}[b]{0.24\textwidth}
        \vspace{0.2cm}
        \includegraphics[width=\textwidth]{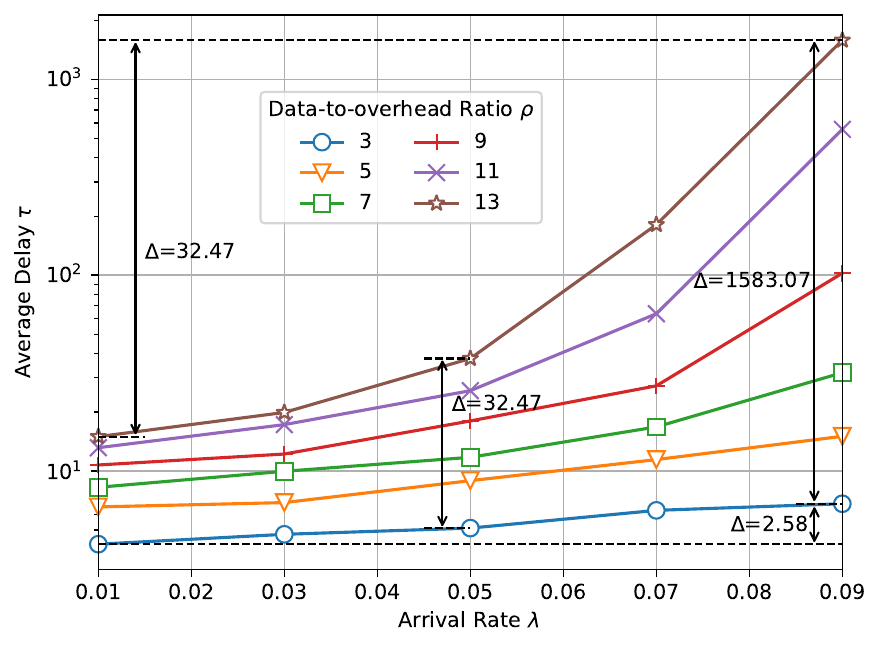}
        \vspace{-0.6cm}
        \caption{Average delay $\tau$ when the data-to-overhead ratio varies.}
        \label{fig:RatioDelay}
    \end{subfigure}
    \caption{Performance of the proposed reservation protocol under various system settings. The dotted lines represent the maximum effective throughput (i.e., $\lambda\rho$), and some data points may exceed the maximum due to stochastic data arrival.}
    \label{fig:Performance}
    \vspace{-1em}
\end{figure*}

\begin{figure*}[!t]
    \centering
    \begin{subfigure}[b]{0.24\textwidth}
        \includegraphics[width=\textwidth]{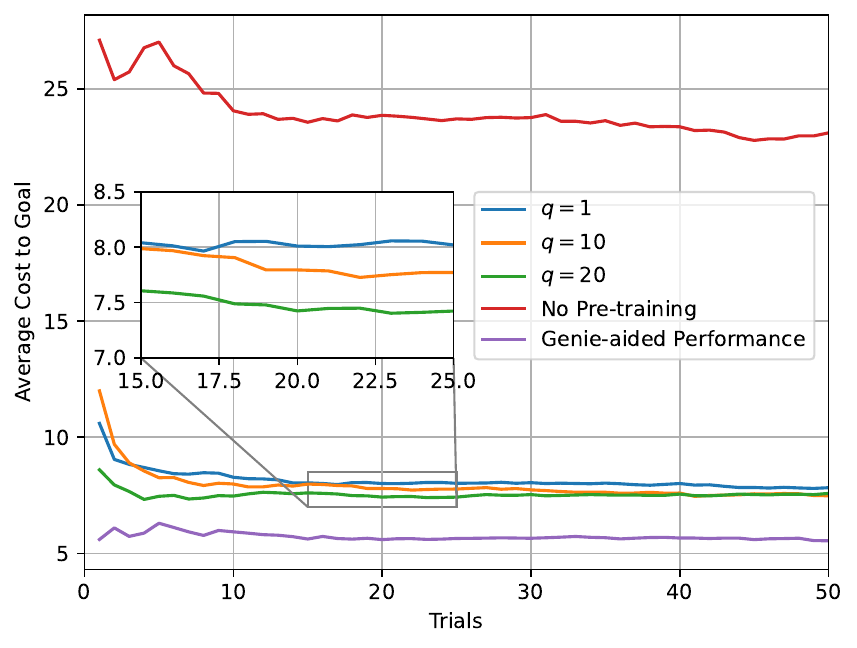}
        \vspace{-0.6cm}
        \caption{Short term behavior of RTDP-Bel under various $q$.}
        \label{fig:qShortterm}
    \end{subfigure}
    \hfill
    \begin{subfigure}[b]{0.24\textwidth}
        \includegraphics[width=\textwidth]{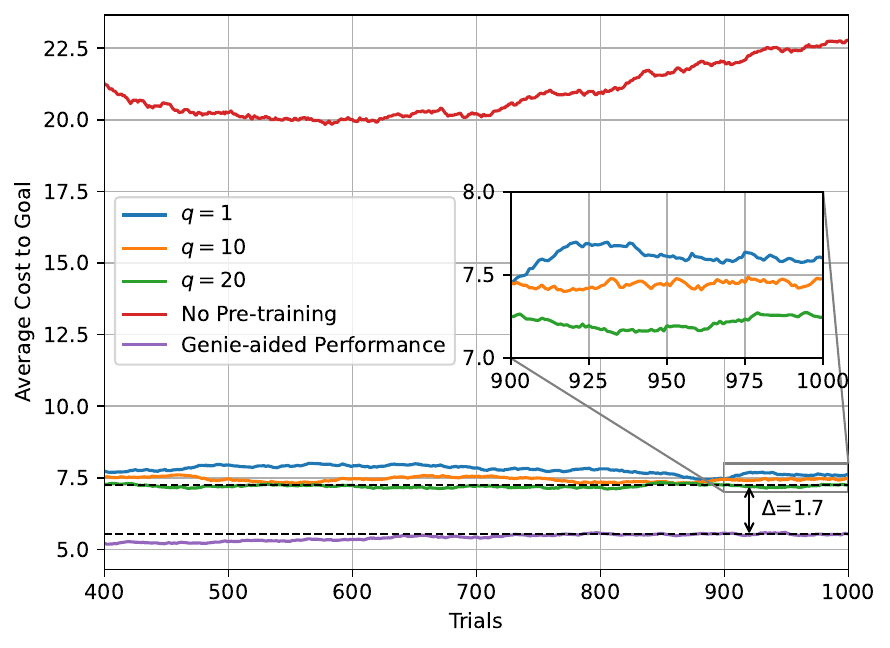}
        \vspace{-0.6cm}
        \caption{Long term behavior of RTDP-Bel under various $q$.}
        \label{fig:qLongterm}
    \end{subfigure}
    \hfill
    \begin{subfigure}[b]{0.24\textwidth}
        \vspace{0.2cm}
        \includegraphics[width=\textwidth]{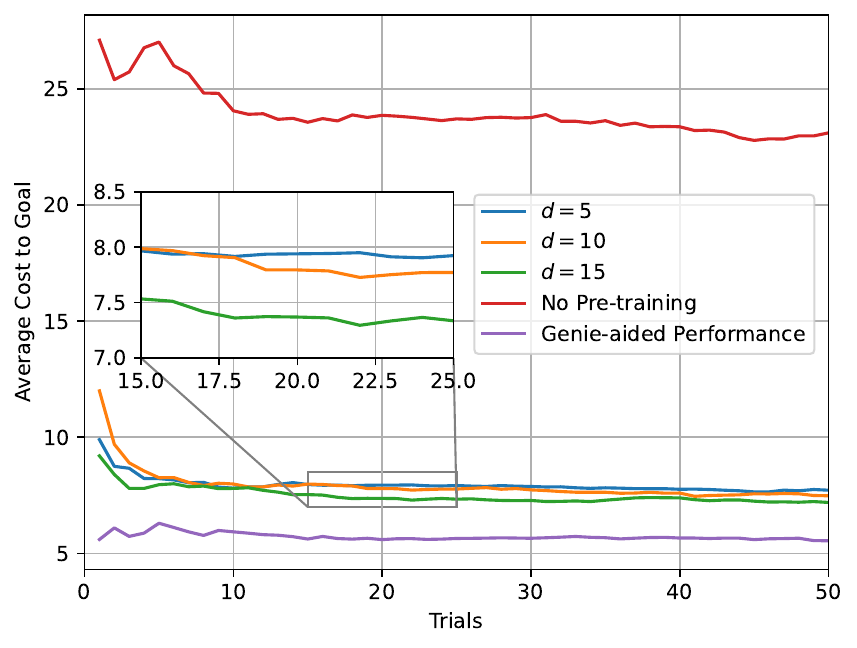}
        \vspace{-0.6cm}
        \caption{Short term behavior of RTDP-Bel under various $d$.}
        \label{fig:dShortterm}
    \end{subfigure}
    \hfill
    \begin{subfigure}[b]{0.24\textwidth}
        \vspace{0.2cm}
        \includegraphics[width=\textwidth]{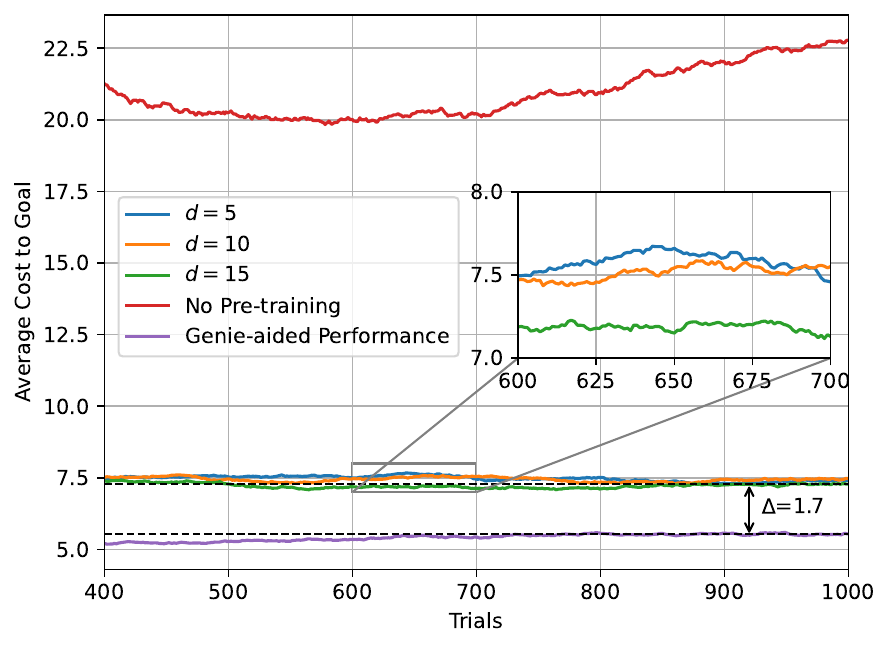}
        \vspace{-0.6cm}
        \caption{Long term behavior of RTDP-Bel under various $d$.}
        \label{fig:dLongterm}
    \end{subfigure}
    \caption{Impact of RTDP-Bel's parameters on its performance, where the short-term behavior is obtained by averaging the moving average with a window size of $40$ trials over $10$ runs and the long-term behavior is represented by a moving average with a window size of $400$ trials from a single run. The "No Pre-training" behavior corresponds to initializing $\tilde{V}(b)$ as zero, while the "Genie-aided Performance" behavior represents the performance of the optimal policy for genie-aided MDP.}
    \label{fig:Parameter}
    \vspace{-1.5em}
\end{figure*}

In the sequel, we study the performance of the proposed reservation protocol under various system settings, as visualized in Fig.~\ref{fig:Performance}. Fig.~\ref{fig:FrameThroughput} shows that under light traffic conditions, all frame division strategies achieve similar performance. However, as traffic intensity increases, performance degrades when the frame length is fixed. For short frames, the degradation is primarily due to the frequent reservation process, which consumes excessive channel resources. For long frames, the performance degradation primarily results from inefficient channel utilization, as the number of idle slots increases. The dynamic frame division strategy, on the other hand, performs better under heavy traffic, as reservations are initiated immediately after data transmissions complete. The average delay under different frame division strategies is shown in Fig.~\ref{fig:FrameDelay}. As we can see, for fixed frames, the delay dramatically increases when the arrival rate exceeds a certain threshold. This occurs because overly frequent reservations in a fixed frame lead to severe data backlog. In contrast, for the dynamic frame, the delay remains lower and stable, as data transmission proceeds uninterrupted by the reservation process. It is also worth noting that the proposed reservation protocol ensures FIFO service. This is in contrast to the benchmark protocols considered in Fig.~\ref{fig:Benchmarks}, where the delivery order is not guaranteed. The effective throughput under different $\rho$ is shown in Fig.~\ref{fig:RatioThroughput}. As shown, the proposed reservation protocol achieves effective throughput very close to the upper limit. Note that some data points exceed the upper limit due to the stochastic data arrival. The average delay when $\rho$ varies is shown in Fig.~\ref{fig:RatioDelay}. As expected, the delay increases with $\rho$. This is because larger data packets require longer transmission times. Additionally, for a fixed $\rho$, an increase in the arrival rate leads to longer waiting times due to the FIFO service provided by the proposed reservation protocol.

Finally, we investigate the impact of RTDP-Bel's parameters on its performance. The results are shown in Fig.~\ref{fig:Parameter}. In these simulations, we run Algorithm~\ref{Alg:RTDP-BelSingleTrial} for multiple trials and omit stochastic data arrival. Hence, we omit the data transmission process and fix the initial belief state as $b_0=[0.1,0.1,0.3,0.3,0.2]$. The first parameter we investigate is the quantization parameter $q$. From Fig.~\ref{fig:qShortterm}, we can see that a small $q$ leads to higher cost due to the merging of distinct states with different optimal actions. However, increasing $q$ from $1$ to $20$ enlarges the size of the hash table from $1194$ to $3751$ entries, but the performance improvement is marginal. A similar trend is also observed in Fig.~\ref{fig:qLongterm}, which presents the long-term behavior. Next, we examine the effect of the discretization parameter $d$. The short-term and long-term behavior are shown in Fig.~\ref{fig:dShortterm} and Fig.~\ref{fig:dLongterm}, respectively. We observe that a small $d$ leads to performance degradation, while increasing $d$ improves performance. However, the performance gain is relatively small compared to the significant increase in computational complexity. For instance, increasing $d$ from $5$ to $15$ results in an eightfold increase in the size of the decision variable. Fig.~\ref{fig:Parameter} also highlights the benefits of pre-training. Specifically, the "No Pre-training" case exhibits high costs, whereas the proposed reservation protocol with pre-training exhibits reduced costs. Additionally, in long-term behavior, the proposed reservation protocol with $d=15$ achieves an average cost of around $7.1$, compared to $5.4$ for the optimal genie-aided policy. This performance gap arises from the lack of knowledge about the number of active terminals. Hence, the extra cost is required to confirm that there are no active terminals. Since this extra cost is around $1.7$, and confirmation incurs a cost of at least $1$, the proposed reservation protocol performs close to the optimal genie-aided policy.

\section{Conclusions}
This paper investigates the random multiple access problem and proposes an efficient reservation protocol optimized through learning under the POMDP framework. We model the system dynamics under a general class of reservation protocols as a POMDP and exploit its structural properties to develop an optimized reservation protocol. To this end, we employ RTDP-Bel, which efficiently learns to optimize channel reservations while maintaining computational tractability. Additionally, we introduce a pre-training strategy that leverages results from a genie-aided problem to accelerate learning. The proposed reservation protocol also ensures FIFO service, guaranteeing that packets arriving within the current frame are always serviced before those in subsequent frames. Unlike conventional random access protocols, the structured transmission approach in the proposed reservation protocol enables a more tractable analysis of system performance, including stability, throughput, and delay. Finally, numerical results demonstrate the superiority of our reservation protocol over both classic and modern random access protocols, with performance gains primarily attributed to the shorter reservation process, which leads to higher effective throughput and improved efficiency, particularly under heavy traffic conditions.

%\appendices
%\section{Appendix}

\bibliographystyle{IEEEtran}
\bibliography{mybib}

\end{document}